\documentclass{cccg25}
\usepackage{graphicx,amssymb,amsmath}
\usepackage[]{todonotes} 
\usepackage{gensymb}

\usepackage{caption}
\usepackage{subcaption}
\usepackage[nocompress]{cite}
\usepackage{hyperref}

\title{Optimal Parallel Algorithms for Convex Hulls in 2D and 3D under Noisy Primitive Operations}

\author{Michael T. Goodrich\thanks{University of California, Irvine, \texttt{goodrich@uci.edu}. Research supported in part by NSF grant CCF-2212129.}
	\and
	Vinesh Sridhar\thanks{University of California, Irvine, \texttt{vineshs1@uci.edu}.}}

\index{Goodrich, Michael T.}
\index{Sridhar, Vinesh}

\begin{document}
\thispagestyle{empty}
\maketitle

\begin{abstract}
  In the noisy primitives model, each primitive comparison performed
  by an algorithm, e.g., testing whether one value is greater than
  another, returns the incorrect answer with random, independent
  probability $p < 1/2$ and otherwise returns a correct answer.
  This model was first applied in the context of sorting and
  searching, and recent work by Eppstein, Goodrich, and Sridhar
  extends this model to sequential 
  algorithms involving geometric primitives such as orientation
  and sidedness tests. 
  However, their approaches appear to be inherently sequential; hence,
  in this paper,
  we study parallel computational geometry algorithms for 2D and 3D
  convex hulls in the noisy primitives model.
  We give the first optimal parallel algorithms in the noisy
  primitives model for 2D and 3D convex hulls in the CREW PRAM model. The
  main technical contribution of our work concerns our ability to
  detect and fix errors during intermediate steps of our algorithm
  using a generalization of the failure sweeping technique.
\end{abstract}

\section{Introduction}
Fault tolerance is an important issue in parallel algorithm design,
since parallelism naturally introduces more opportunities for faults to occur.
In addition, applications can involve
the use of faulty primitive operations, such as comparisons.
It
is ideal if our algorithms can tolerate such noisy primitives.
Accordingly,
in this paper, we are interested in the design of efficient parallel
computational geometry
algorithms where primitive operations (such as orientation tests) 
can return the wrong answer with a fixed probability,
$p<1/2$.

Motivation for this model of noise comes, for example,
from potential applications involving quantum computing.
A quantum computer is used to answer primitive queries, 
which return an incorrect
answer with a fixed probability at most $p<1/2$; 
see, e.g.,~\cite{allcock2023quantumtimecomplexitydivide,iwama,math11224707}. 
A second motivation for this noise model comes from distributed computing
applications, in which hash 
functions can be used to reduce communication 
costs at the expense of introducing comparison errors, such as in work
by Viola~\cite{Vio-Comb-15}.
For example, in one problem studied by Viola, two participants with
$n$-bit values seek to determine which of their two values is
largest. This can be done by a noisy binary search for the highest-order
differing bit position. Each search step performs a noisy equality
test on two prefixes of the inputs, by exchanging single-bit hash
values on prefixes. The result is an $O(\log n)$ bound on the randomized
communication complexity of the problem. 

A simple observation, made for sorting by 
Feige, Raghavan, Peleg, and Upfal~\cite{feige1994computing},
is that we can use any parallel algorithm based on correct
primitives by performing each primitive operation 
$O(\log n)$ times (either sequentially or in parallel) 
and returning the majority answer.  
This guarantees correctness w.h.p., 
but increases the span or work of the parallel algorithm
by a $\log n$ factor. 
In this paper, we are interested in the design
of
parallel computational geometry algorithms with noisy primitive operations
that are correct w.h.p. without incurring this overhead. 

\paragraph{Related Prior Work.}
The model for fault tolerance has been extensively studied for 
binary comparisons for 1-dimensional values. 
However, the only prior work we are aware of
studying fault-tolerant computational 
geometry algorithms does so in the sequential setting. 
A recent work by Eppstein, Goodrich,
and Sridhar~\cite{eppstein2025computational} use a technique based on random walks
  in DAGs to develop optimal sequential algorithms for several
  computational geometry problems in this setting. 
 Their methods seem to be inherently sequential, however,
in that their
  2D convex hull algorithm uses a noise-tolerant version of
the Graham scan algorithm, and their sequential 3D convex
  hull algorithm applies their random walks technique to a randomized
  incremental construction (RIC) algorithm, which seems to be a poor candidate
for the RIC parallelization techniques of 
Blelloch, Gu, Shun, and Sun~\cite{blelloch2020parallelism,blelloch2020randomized}, 
due to their methods needing correct primitive operations.
Thus, to achieve efficient parallel algorithms in the noisy primitives
model, new approaches are needed.

Work on searching and sorting 1-dimensional data
with faulty comparisons originates with
R{\'e}nyi~\cite{renyi}, who studied a binary search problem
in 1961
where comparisons between two values return the wrong answer independently with
probability $p<1/2$; see also, e.g., Pelc~\cite{PELC1989185,PELC200271}.
In 1994, 
Feige, Raghavan, Peleg, and Upfal~\cite{feige1994computing} showed
how to sort $n$ items in $O(\log n)$ span with $O(n\log n)$ work in the
CRCW PRAM model, with high probability,\footnote{
  In this paper, we take ``with high probability'' (w.h.p.) to mean that the failure probability
  is at most $1/n^c$, for some constant $c\ge 1$.}
in the same noisy comparison model,
and
Leighton, Ma, and Plaxton~\cite{LEIGHTON1997265} show how to achieve
these bounds in the EREW PRAM model
(we review the PRAM model of computation in Appendix~\ref{sec:parallel-review}).
Both algorithms make extensive use of the
AKS sorting network~\cite{aks1,aks2},
however, which implies that their asymptotic performance bounds
have very large constant factors.
Recently, in a SPAA 2023 paper, 
Goodrich and Jacob~\cite{goodrich2023optimal} 
provide 
an efficient parallel sorting algorithm where comparisons
are either persistently or independently faulty
without using 
the AKS sorting network. 
Since computations cannot be fully correct with persistent errors,
we focus in this paper
on the design of efficient parallel computational geometry algorithms
using Boolean geometric primitives, 
such as orientation queries or sidedness tests, that randomly return
the wrong answer independently of previous queries with probability at most $p<1/2$ (where $p$ is known).


\paragraph{Our Results.}
In this paper, we give efficient parallel algorithms for constructing
convex hulls of $n$ points in 2D and 3D when primitive orientation
tests return the wrong answer with probability $p<1/2$. We give 
algorithms in the CREW PRAM model that have $\Theta(\log n)$
span and $\Theta(n\log n)$ work, w.h.p. The lower bounds follow from lower bounds on computing the minimum/maximum of a set of $n$ points~\cite{gupta1997optimal,feige1994computing}.

Our algorithms involve parallel divide-and-conquer and novel versions of the 
\emph{failure sweeping} technique~\cite{matias1991converting,goodrich2023optimal,feige1994computing,ghouse1991place},
which is potentially of independent interest. 
The failure sweeping technique is a paradigm for increasing the 
success probability of a divide-and-conquer parallel algorithm.
Unfortunately, when a success probability depends on the input size, the
success probability of recursive calls can be too low.
The solution, then, is to design an oracle that can test w.h.p. the solutions
to the recursive subproblems, and then ``sweep'' failed subproblems
into a memory space for which we can then apply additional parallel resources
to solve.
The trick is to show that the total 
size of all such failure subproblems is relatively small
so that we can then solve them using
a parallel algorithm that would otherwise have an inefficient work bound.
In this paper, we show how to adapt the failure sweeping technique to
the case of probabilistically noisy geometric primitives.

We believe that our
  techniques can be adapted to develop parallel algorithms
for other computational geometry problems, such as
  visibility, point-location, etc., in the noisy primitive
  setting, and more generally to develop parallel divide-and-conquer
  algorithms where errors are associated with both construction and
  validation steps.

\section{Preliminaries}
\paragraph{Noisy Geometric Primitives.}
Convex hull algorithms, as well as other geometric algorithms typically rely on primitive operations, such as orientation tests or visibility tests, that output a Boolean value and operate in constant time. In this work, we assume that each such primitive operation has a fixed probability $p < 1/2$ of being incorrect. We assume \emph{non-persistent} errors, meaning that each operation can be modeled as an independent weighted coin flip. 
As described above, Feige {\it et al.}~\cite{feige1994computing}, observed that we can convert an algorithm that assumes correct primitives into a noise-tolerate one by repeating each noisy primitive operation $O(\log n)$ times and choosing the majority response. Eppstein, Goodrich, and Sridhar~\cite{eppstein2025computational} call this the \emph{trivial repetition strategy}, and it immediately implies $O(\log^2 n)$-span, $O(n\log^2 n)$-work algorithms for convex hull construction w.h.p., as well as for a variety of other problems. Our goal is to judiciously use this strategy alongside other noise-tolerant algorithms and our novel failure sweeping techniques to improve this to optimal span and work. 

As Eppstein {\it et al.} note, manipulating non-geometric data, such as pointers and metadata, or manipulating already-processed data is not a noisy operation. For example, rebalancing a tree used as a point-location structure for a planar convex hull is not noisy. However, querying that same structure is a noisy operation. Another important non-noisy operation is the prefix sum. In our paper, we use prefix sums to combine array elements either arithmetically or via concatenation. As a result, no geometric primitives are used, and so they are non-noisy. 
\paragraph{Path-guided Pushdown Random Walks.}
The main technical contribution of Eppstein, Goodrich, and Sridhar is called \emph{path-guided pushdown random walks}. This technique is a generalization of the noisy binary search of~\cite{renyi} to apply to any DAG that obeys the following two constraints: (1) each possible query value must be associated with a unique path $P$ down the given DAG and (2) we can verify whether we are on a correct path using a constant number of noisy comparisons. 

As a simple example, consider a binary search tree. Constraint (1) is satisfied by properties of binary search. To satisfy constraint (2), we can maintain pointers to two ancestors $v_l$ and $v_r$ at each node $v$. They are respectively the largest ancestor of $v$ whose value is smaller than $v$'s and the smallest ancestor of $v$ whose value is larger than $v$'s (one of these will be $v$'s parent). The values of $v_l$ and $v_r$ bound all possible values that $v$ and its descendants can take, so in two comparisons we can determine if a query $q$ could be located in $v$'s subtree, i.e., if we are going down the correct path. 

In general, the authors call this second step devising a \emph{transition oracle} and do so for a variety of DAGs used in the construction of several geometric data structures. They show that, if such a transition oracle exists, then in $\Theta(|P| + \log(1/\varepsilon))$ steps we find our desired answer or determine a non-answer correctly with success probability $1 - \varepsilon$. If the height of our DAG is $O(\log n)$ and our confidence bound is w.h.p. in $n$, this does not increase search runtimes asymptotically. Technical details for path-guided pushdown random walks can be found in Appendix~\ref{sec:random-walks}. 
\section{Failure Sweeping with Noisy Primitives}\label{sec:failure-sweep-general}
Failure sweeping, formalized by Ghouse and Goodrich~\cite{ghouse1991place}, is a general technique used to improve confidence bounds on algorithms that either fail or exceed a desired span or work bound with some failure probability $p(n)$, given input size $n$. Say that our desired bounds for span and work done at a subproblem of size $n$ are $T(n)$ and $W(n)$. Our goal is for the entire algorithm to fail with probability at most $p(n)$, but when we recurse to sufficiently small subproblems of size $m$, e.g., $m = O(\log n)$, $p(m)$ may be exponentially larger than $p(n)$. If unchecked, errors in small subproblems propagate to the returned answer, meaning the algorithm can only guarantee success probability $1 - p(m)$, rather than the intended $1 - p(n)$. This is exactly the case in the noisy primitives setting. Algorithms resilient to noise run in span $T(m)$ if we wish for them to fail with probability at most $p(m)$. Attempting to have all subproblems fail with probability at most $p(n)$ would cause each subproblem to take span proportional to $T(n)$, leading to an inefficient algorithm. Failure sweeping gives us a way to improve this confidence bound from $1 - p(m)$ to $1 - p(n)$ without incurring this span blow-up.  

Let $n$ be the size of a parent subproblem and $m$ the size of its child subproblems. The general failure sweeping framework introduces two algorithms: a failure sweeping algorithm that can determine whether one of $n$'s subproblems of size $m$ is incorrect and a ``brute-force'' algorithm that can recompute an incorrect subproblem w.h.p. in $n$. We must sweep all of $n$'s subproblems at every level of recursion, so failure sweeping over all subproblems must not take more than $T(n)$ span and $W(n)$ work. On the other hand, our brute-force algorithm only applies to incorrect subproblems. Thus, brute-force recomputation of a single problem of size $m$ may exceed $T(m)$ and $W(m)$, though brute-force recomputation over all failed subproblems may not exceed $T(n)$ and $W(n)$ w.h.p. in $n$.

To complete the argument, one must show that not too many subproblems fail with probability $p(n)$. Specifically, if each subproblem of size $m$ fails with probability $p(m)$, we must show that no more than $q$ subproblems, determined in the analysis, fail with probability $1 - p(n)$. Then, we apply our brute-force algorithm to compute the $q$ subproblems in span $T(n)$ and total work $\leq W(n)$.
This way, all subproblems succeed with probability $1 - p(n)$. Applied inductively, this shows that we can ``upgrade'' our failure probabilities at each level of recursion without asymptotically increasing span or work. This ultimately allows the algorithm to achieve confidence bound $1 - p(n_0)$, where $n_0$ is the initial problem size.

One challenge in designing failure sweep algorithms is in ensuring that $m$ is not too small and not too large compared to $n$. If $m$ is too small, then many subproblems are likely to fail, forcing us to recompute too many subproblems. If $m$ is too large, then the cost to recompute even a few subproblems may exceed $W(n)$. Indeed, this second issue arises in our 3D convex hull algorithm. We introduce a novel technique called \emph{size reduction}, which is a preprocessing step that transforms large subproblems of size $m$ into $O(m)$ sufficiently small problems of total size $O(m)$ with probability $1 - p(n)$. We are then able to show that not too many of these smaller subproblems fail, guaranteeing that our total brute-force work is no more than $W(n)$. 
How small these subproblems must be depends both on $W(n)$, the brute-force algorithm, and $p(m)$.

In the case where 
$p(n)$ is a high-probability bound in $n$, 
size reduction may be applied to any algorithm 
wherein this issue arises as long as one can efficiently transform a large subproblem into $O(m^\alpha)$ sufficiently small subproblems of total size $O(m)$ with probability $1 - p(n)$ for any fixed constant $\alpha$.
Higher confidence guarantees may tolerate a larger fanout. 

\section{A Deterministic 2D Convex Hull Algorithm for CREW PRAM}\label{sec:our-2D-hull}

\paragraph{Outline of Our Algorithm.}
Our algorithm closely follows the $O(\log n)$-span, $O(n\log n)$-work algorithm of~\cite{aggarwal1988parallel}~and~\cite{goodrich1987efficient}, which we review in Appendix~\ref{sec:2dhull-original}. However, we swap operations subject to noise with noise-tolerant replacements. The initial sorting operation is replaced with the parallel noisy sorting of \cite{goodrich2023optimal}. 
Upper tangents are computed using the path-guided pushdown random walks technique of Eppstein, Goodrich, and Sridhar~\cite{eppstein2025computational}, which we describe in Appendix~\ref{sec:upper-tangent-oracle}. The max find operation used to compute the $V_j$'s and $W_j$'s for each upper hull is replaced with the noise-tolerant version described by Feige, Raghavan, Peleg, and Upfal~\cite{feige1994computing} for EREW PRAM. Note that the final parallel prefix operation is not a noisy operation because it only requires us to compare the indices 1 through $\sqrt n$ that we assigned each upper hull after sorting. Also, we note that the base case, when $n \leq 2$, can be computed without noisy comparisons as we simply link the two points with an edge. 

Each replacement matches the work and span of their non-noisy counterparts. 
However, they are only correct with high probability in the size of the current subproblem.
At small problem sizes, unchecked errors will compound such that we fail to produce an upper hull of all $n$ points with high probability in $n$. We apply failure sweeping
to detect and fix these errors. As described in Section~\ref{sec:failure-sweep-general}, we will develop a failure sweep and brute-force algorithm and use them to show that the entire algorithm succeeds with high probability in $n$ with optimal span and work. 

\paragraph{Failure Sweeping.}\label{sec:2D-failure-sweep}
In this section, we will use failure sweeping to show that we can certify and recompute $m$ subproblems of size $m$ in $O(\log(m^2))$ span and $O(m^2\log(m^2))$ work with high probability in $m^2$.
This way, we upgrade our failure probability to the size of the recursive step to which we return.
For now, we will describe individual noisy operations as failing with probability at most $1/m^c$. Later, we show what value $c$ must take. 

Consider a subproblem $S$ of size $m$. We will show how to detect whether $S$'s upper hull is valid. Some of the points in this subproblem belong to the returned upper hull and some do not (we can label each point with a corresponding status). Of the points that belong to the hull, they perform orientation tests on their neighbors using the trivial repetition strategy such that they fail with probability at most $m^{-2c}$. This verifies that the hull is convex. For the points not on the hull, we can apply noisy binary search
to determine the edge of the upper hull directly above or below them (recall that upper hulls are represented in memory as balanced binary search tree of their points sorted by $x$-coordinate). Adjusting constant factors will allow us to compute each search with failure probability at most $m^{-2c}$. A final orientation test can determine whether each non-hull point is below the upper hull---valid---or above the upper hull---invalid. This takes $O(\log(m^2))$ time using the trivial repetition test. In total, testing all $m$ points takes $O(\log(m^2))$ span and $O(m\log(m^2))$ work.
Every point reports whether it is correct, and we can use a prefix operation to determine whether $S$ is correct in an additional $O(\log m)$ span and $O(m)$ work. If $S$ is found to be incorrect, we use a brute-force algorithm to recompute it. See Figure~\ref{fig:failure-sweep-2dch} for an example of failure sweeping on some subproblem $S$.

\begin{figure}
    \centering
    \includegraphics[width=0.8\linewidth]{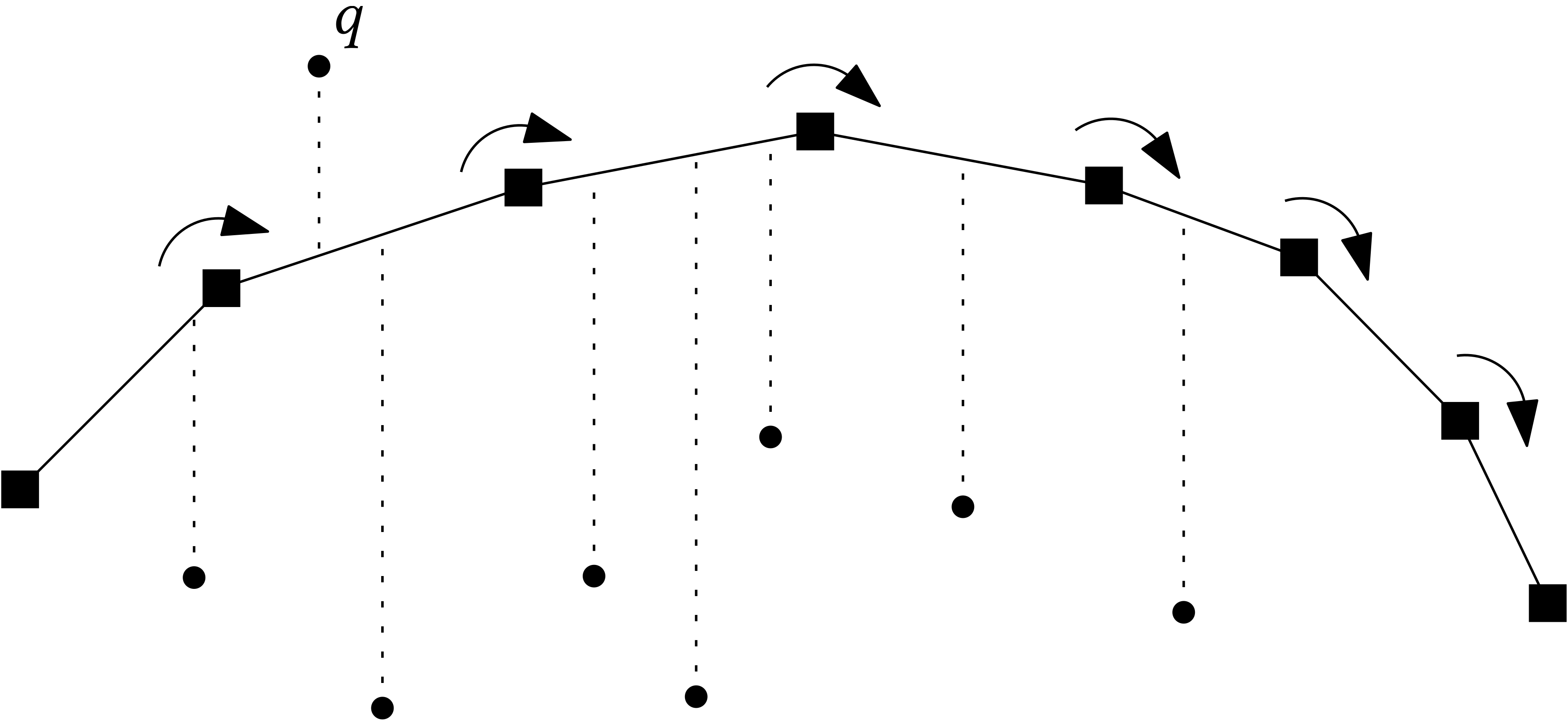}
    \caption{Here we have an example of failure sweeping on some upper hull of a subproblem $S$. To verify that each (square) hull point is correct, we use the trivial repetition strategy to determine that it is convex with respect to its left and right neighbors. We perform a noisy binary search on the hull to determine if any non-hull points are outside the upper hull. This is true for point $q$, so this subproblem is flagged for reconstruction.}
    \label{fig:failure-sweep-2dch}
\end{figure}

\paragraph{Brute-force Recomputation.}
We will show that we can recompute with brute-force a subproblem $S$ of size $m$ with high probability in $m^2$ using $O(\log(m^2))$ span and $O(m^{3/2}\log(m^2))$ work. 
We again begin by splitting $S$ into $\sqrt m$ groups $S_j$, each of size $O(\sqrt m)$. However, we will solve each $S_j$ directly, rather than through recursion. 

Let us assume for simplicity that $|S_j| > 6$. First, compute the topmost, leftmost, and rightmost points in each $S_j$, and call them $T_j$, $R_j$, and $L_j$ respectively. $L_j$ and $R_j$ can be recovered in constant time from the fact that the points have already been sorted by their $x$-coordinate.
$T_j$ can be determined using the noisy max-find of Feige {\it et al.}~\cite{feige1994computing} in $O(\sqrt m\log(m^2))$ work and $O(\log(m^2))$ span with failure probability at most $m^{-2c}$. These three points must be on the upper hull of $S_j$ as they are extreme points in their respective directions. For every other point $p \in S_j$, we find $p$'s clockwise neighbor, i.e., the point $r$ such that the angle between $\overline{pr}$ and a vertical line through $p$ is minimized (this is similar to the ``gift-wrapping'' process used in the Jarvis march convex hull algorithm~\cite{jarvis1973identification}). Using the noisy max-find of Feige {\it et al.}~\cite{feige1994computing} and adjusting constant factors, each max-find fails with probability at most $m^{-2c}$. We do a similar process to find $p$'s counterclockwise neighbor. Call $p$'s CCW neighbor $q$ and its CW neighbor $r$. 

We can then use $T_j, R_j, L_j, q, p,$ and $r$ to determine whether $p$ is on the upper hull of $S_j$. With a constant number of orientation tests, we can determine whether either sequence of vertices $[L_j, q, p, r, T_j, R_j]$ or $[T_j, q, p, r, R_j, L_j]$ is a convex chain. Using the trivial repetition strategy, this takes $O(\log(m^2))$ time to guarantee with failure probability at most $m^{-2c}$. See Figure~\ref{fig:brute-force-2dch} for a visual of the following lemma.

\begin{lemma}\label{lem:2D-extreme}
    If either $[L_j, q, p, r, T_j, R_j]$ or $[T_j, q, p, r, R_j, L_j]$ is a convex chain, then $p$ is a member of the upper hull of $S_j$. 
\end{lemma}

\begin{proof}
See Appendix~\ref{sec:2D-supp-2}.
\end{proof}

To complete the upper hull for $S_j$, we simply perform a parallel prefix operation on the points, ordered by their $x$-coordinate. Points that were found to not be on the hull contribute nothing to the prefix operation. Recall that $S_j$ is just one of $\sqrt m$ upper hulls in the subproblem $S$ that we wished to recompute. To compute the upper hull over all points in $S$, we simply use the original recombination method of~\cite{goodrich1987efficient,aggarwal1988parallel}, with the corresponding noise-tolerant upper tangent and max-find algorithms described above. Again, parameters can be set such that each of those operations fails with probability no more than $m^{-2c}$ with only a constant factor increase in runtime. 

It is clear that computing upper hulls for all $S_j$ takes $O(\log(m^2))$ span, as each point $p\in S_j$ can perform its own computations independently in $O(\log(m^2))$ time and all $S_j$ can be processed independently of each other. The work to compute the upper hull of $S_j$ associated with $p$ is performing two max-finds on $O(\sqrt m)$ points to retrieve $q$ and $r$ plus the work to perform the orientation tests to verify whether $p$ is a hull point or not. This takes $O(\sqrt m \times O(\log(m^2)))$ work. Then over all $m$ points, total work is $O(m^{3/2} \log(m^2))$. A constant number of noisy operations are performed per point, so this takes $O(m)$ noisy operations in total.

The remaining operations, a max-find to find $T_j$ for each $S_j$, a parallel prefix operation to create the upper hull for $S_j$, and the recombination step for each $S_j$ all take $O(\log(m^2))$ span and $O(m\log(m^2))$ work in total, using $O(m)$ noisy operations.

\begin{lemma}\label{lem:2D-failure-sweep}
    We can failure-sweep an upper hull of $m$ points in  $O(\log(m^2))$ span and $O(m\log(m^2))$ work using $O(m)$ noisy operations. 
We can compute an upper hull in $O(\log(m^2))$ span and $O(m^{3/2}\log(m^2))$ work using $O(m)$ noisy operations. Both hold w.h.p.~in $m^2$.
\end{lemma}


\paragraph{High-Probability Bounds for Our Convex Hull Algorithm.}
\label{sec:2D-inductive}
Our subproblem $S$ of size $m$ returns to a subproblem $S'$ of size $m^2$. To maintain the same work recurrence as the original algorithm, subproblem $S'$ must limit its work to $O(m^2\log(m^2))$. We have the ``budget'' to sweep each subproblem $S$. However, since recomputing a single subproblem takes $O(m^{3/2}\log(m^2))$ work, recomputing all $m$ subproblems of size $m$ would take $O(m^{5/2}\log(m^2))$ work. 

With some arithmetic, we can bound the number of subproblems $S$ that can fail such the we can afford to recompute them all by brute-force. 

\begin{cor}\label{lem:2D-bound-fail}
    During the recombination step of subproblem $S'$, at most $O(m^{1/3})$ subproblems can fail such that we take $O(m^2\log(m^2))$ work to recompute them. 
\end{cor}
\begin{proof}
See Appendix~\ref{sec:2D-supp-2}.
\end{proof}

Applying our failure-sweeping and brute-force algorithms inductively gives us the following.

\begin{theorem}
    We can compute the convex hull of $n$ points in the plane in $O(\log n)$ span and $O(n\log n)$ work w.h.p. in $n$, even with noisy primitive operations.
\label{thm:2dc}
\end{theorem}
\begin{proof}
See Appendix~\ref{sec:2D-supp}.
\end{proof}

\section{A Randomized 3D Convex Hull Algorithm for CREW PRAM}

See Appendix~\ref{sec:3dhull-original} for a review of the non-noisy algorithm of Reif and Sen~\cite{reif19923DHull}. 
As before, there are three steps to our analysis. We first show how to adapt all operations of the original algorithm such that they are noise-tolerant w.r.t. their input size. Then we show how to efficiently perform failure-sweeping and brute-force reconstruction. The first two steps are intricate but straightforward, primarily using old techniques. We defer details to Appendix~\ref{sec:3D-supp} and, in particular, prove the following.
\begin{lemma}\label{lem:3D-brute-force}
    The intersection of a set of $n$ halfspaces can be computed in $O(\log n)$ span and $O(n^4\log n)$ work with high probability in $n$ under noisy comparisons.
\end{lemma}
\begin{theorem}\label{lem:failure-sweep-3D}
    Any parent subproblem of size $M$ can failure sweep all of its child subproblems in $O(M\log M)$ work and $O(\log M)$ span w.h.p. in $M$. This uses $O(M)$ noisy operations. 
\end{theorem}
Here, we focus on brute-force construction, for which we  apply our generalization of the failure sweep paradigm.

\paragraph{Efficient Brute-Force Reconstruction via Size Reduction.}
Attempting to apply the strategy that bounds the number of failed subproblems that we used for 2D convex hulls presents a problem. Recall from Reif and Sen's original work~\cite{reif19923DHull} as well as Clarkson~\cite{clarkson1988applications}, if the problem at recursive level $l-1$ is size $M$, then its child problems at recursive level $l$ could be size at most $O(M^{7/8}\log M)$. Indeed, if we were to apply Lemma~\ref{lem:3D-brute-force} to solve just a single large subproblem by brute-force, it would take $O(M^{7/2}\log M)$ work to solve it w.h.p. in $M$. In this section, we develop a novel technique we call \emph{size reduction} to decompose these larger subproblems into problems that are small enough to recompute if needed but are large enough to have similar failure probabilities.



To illustrate our approach to size reduction, consider the following algorithm to compute a halfspace intersection of $m$ halfspaces w.h.p. in $M$. We can perform the original algorithm with the following two modifications. (1) Each recursive level must succeed with high probability in $M$, so no matter what the current problem size is we tune our noisy operations and sampling probabilities to fail with probability $\leq 1/M$. (2) Irrespective of the current problem size, always choose a sample of size $O(m^{1/16})$. 

By Reif and Sen's original work~\cite{reif19923DHull} as well as an analysis of Clarkson~\cite{clarkson1988applications}, 
each subproblem of size $c$ now takes $O(\log M)$ span and $O(\sqrt m\log M + c\log M)$ work to preprocess it into recursive subproblems.
However, because our samples are larger, we reduce the maximum subproblem size by roughly $O(m^{1/16})$ at each level, 
meaning that it takes at most fifteen levels of recursion such that the largest problem size is $O(m^{1/8})$. When we get to this level, we simply apply Lemma~\ref{lem:3D-brute-force} to solve them w.h.p. in $M$ by brute-force. 

There can be at most $O(m)$ such subproblems 
so, in total, 
the cost of preprocessing and the cost of the final brute-force computation is
at most $O(\log M)$ span and $O(m^{3/2}\log M)$ work. 
All operations are done w.h.p. in $M$, so failure sweeping is unnecessary. We simply recombine them as we return from the fifteen levels of recursion in an additional $O(\log M)$ span and $O(m\log M)$ work. 

\begin{lemma}\label{lem:3D-medium-brute-force}
    The intersection of a set of $m$ halfspaces can be computed in $O(m^{3/2}\log M)$ work and $O(\log M)$ span with high probability in $M$. 
\end{lemma}

We can modify this algorithm to instead operate as a preprocessing step prior to recursion. 
For every subproblem of size $m = \omega(\sqrt M)$, 
perform the above steps for at most nine levels until 
we have decomposed $m$ into at most $O(M^{9/16})$ 
subproblems of size $O(\sqrt M)$ 
(by pruning, the sum of subproblem sizes is no more than $O(M)$). 
This can be performed in $O(\log M)$ span and $O(M\log M)$ work.
Then we recurse on all induced subproblems, tuning parameters such that each succeeds with high probability in $M^{1/2}$. 

If each subproblem is of size $O(M^{1/2})$
and we can recompute a given subproblem of size $m$ by brute force in time $O(m^{3/2}\log M)$,
then at most 
$O(M^{1/4})$ can fail such that we can recompute them all in 
$O(\log M)$ span and $O(M\log M)$ work.
We note that, due to Theorem~\ref{lem:span-thm}, pruning, and 
the structure of the recursion, 
even if a subproblem is of size $o(M^{1/2})$, 
we can still tune parameters to evaluate it with high probability in $M^{1/2}$ without affecting total asymptotic complexity.

\begin{theorem}\label{lem:brute-force-thm}
    We can recompute all failed subproblems w.h.p. in $M$ in $O(\log M)$ span and $O(M\log M)$ work.
\end{theorem}
\begin{proof}
    Follows from Theorem~\ref{lem:failure-sweep-3D}, our analysis above, and an application of Lemma~\ref{lem:3D-medium-brute-force} to any failed subproblems. 
    We will now show that not too many fail w.h.p. in $M$.
    Let $X$ be a random variable representing the number of subproblems that fail and say that each subproblem fails with probability $M^{-c'/2}$ for some constant $c'$ that we will determine.
    Then $E[X] = O(M^{9/16})\times M^{-c'/2} = O(M^{(9 - 8c')/16})$.
    Recall that we require no more than $O(M^{1/4})$ subproblems to fail. 
    Applying Markov's inequality, we have that 
    $$\Pr[X \geq  M^{1/4}] = O(M^{(9 - 8c')/16}) \times M^{-1/4}$$
    $$ = O(M^{(5 - 8c')/16}).$$

    To be a high-probability bound, this failure probability must be less than $M^{-1}$. 
    Then we can recompute all failed subproblems w.h.p. in $M$ so long as $(5 - 8c')/16 < -1$, and so $c' > 21/8$.
    
\end{proof}

Now we are ready to prove that failure sweeping, brute-force recomputation, and size reduction allow us to compute the intersection of $n$ halfspaces in optimal span and work w.h.p. in $n$. 

\begin{theorem}\label{lem:3D-final-thm}
    We can compute the intersection of $n$ halfspaces in $O(\log n)$ span and $O(n\log n)$ work w.h.p. in $n$, even with noisy primitive operations.
\end{theorem}
\begin{proof}
    See Appendix~\ref{sec:3D-supp-pf}. 
\end{proof}

\section{Discussion}
In this paper, we have shown how to use our generalization of the failure sweeping technique to compute convex hulls in 2D and 3D in optimal span and work with high probability in $n$. We believe that our techniques can be adapted to solve more geometric problems in the parallel, noisy primitives model, such as general planar point location, visibility, and convex hulls in higher dimensions. We also hope that our new size reduction technique inspires future work in other settings in which subproblems fail with some probability as a function of their input size. 

\section*{Acknowledgements}
We thank Ryuto Kitagawa for helpful suggestions on writing and organization.


\small
\bibliographystyle{abbrv}
\bibliography{refs,refs2}

\clearpage
\appendix

\section{Review of PRAM and Work-Span Analysis}\label{sec:parallel-review}

In this work, we use the PRAM (Parallel Random Access Machine) model. This parallel model assumes that an algorithm has access to unbounded shared memory which is accessed concurrently at each step of computation. There are variants of the PRAM that specify whether or how processors are allowed to access the same memory location concurrently. The model most prominently used in this work is the CREW (Concurrent Read, Exclusive Write) PRAM, which allows processors to read concurrently but prohibits concurrent writes. Other models we mention are the EREW (Exclusive Read, Exclusive Write) PRAM, which prohibits any concurrent access and the CRCW (Concurrent Read, Concurrent Write), which allows both concurrent reads and writes. 

We use the work-span framework to analyze our parallel algorithms. Assume that a parallel algorithm performs all of its operations in $s$ parallel steps. Parallel steps happen sequentially, i.e., all processors are at the same step at the same time. In each step, each processor may perform $O(1)$ constant-time operations. Let $w_i$ be the number of operations performed in step $i$. We call $s$ the \emph{span} of the algorithm and the total number of operations performed, i.e., $\sum_{i=1}^s w_i$, the \emph{work} of the algorithm. 
\section{Review of Path-guided Pushdown Random Walks}\label{sec:random-walks}
The path-guided pushdown random walks technique was first described by Eppstein, Goodrich, and Sridhar in~\cite{eppstein2025computational}. 
The technique is a generalization of noisy binary search to search on DAGs in which noisy primitives are used to determine which node to traverse to next. 
They require that any query value has a unique valid path $P$ in the DAG. 
At its core, their search is a backtrack search. 
We maintain a stack $S$ of nodes visited and, if we determine we are on the wrong path, our next step is to backtrack by popping the first vertex off of $S$. 
Their unique contribution comes from how to determine if we are on the incorrect path. They model this process using what they call a \emph{transition oracle} $T$. Given a vertex $v$, the job of the transition oracle is to determine the next vertex in the DAG to move to, or to determine whether to backtrack. However, with some fixed probability $p_e < 1/2$, the transition oracle $T$ ``lies'' and returns an arbitrary answer (with the restriction that any vertex it returns must be an outgoing neighbor of $v$ or the vertex at the top of $S$). If the transition oracle acts truthfully and $v=t$, our goal node, then the transition oracle pushes a copy of $t$ to the stack, an extra insurance to prevent us from deviating. 

Assuming that, when $T$ is truthful, it can determine whether we are on the correct path and the next node to advance to using $O(1)$ primitive operations, Eppstein {\it et al.} prove the following theorem.

\begin{theorem}[Theorem 1 of~\cite{eppstein2025computational}]
    Given an error tolerance $\varepsilon = n^{-c}$ for $c > 0$, and a DAG, $G$, with a path, $P$, from a vertex $s$ to a distinct vertex $t$, the path-guided pushdown random walk in $G$ starting from $s$ will terminate at $t$ with probability at least $(1-\varepsilon)$ after $N = \Theta(|P| + \log(1/\varepsilon))$ steps, for a transition oracle, $T$, with error probability $p_e < 1/15$.
\end{theorem}

The basic idea of the proof is that the design of the transition oracle allows each truthful action to undo each action done when $T$ acts arbitrarily. Because $p_e < 1/2$, the good actions outweigh the bad. The authors show that, using a Chernoff bound, we reach our goal node $t$ after $N$ steps and terminate after a further $O(\log(1/\varepsilon))$ steps. The remainder of their work describes how to determine if a query is on the correct path in a DAG in $O(1)$ primitive operations, as this implies that a transition oracle exists for that DAG. This is not trivial to do. For example, it is not obvious how to construct transition oracles for tree-based priority queues. Nevertheless, they were able to describe transition oracles for a variety of geometric DAGs.

\subsection{A Transition Oracle for Computing the Upper Tangent of Two Convex Hulls.}
\label{sec:upper-tangent-oracle}
Our parallel 2D convex hull algorithm requires us to compute several upper tangents between pairs of upper hulls in parallel. The non-noisy convex hull algorithms of~\cite{goodrich1987efficient} and~\cite{aggarwal1988parallel} do this by applying the sequential double binary search technique of Overmars and van Leeuwen~\cite{overmars1981maintenance} to every pair of upper hulls. To make this operation noise-tolerant, we apply the noise-tolerant upper-tangent algorithm developed by Eppstein, Goodrich, and Sridhar~\cite{eppstein2025computational} using path-guided pushdown random walks. 

In the original algorithm of~\cite{overmars1981maintenance}, each upper hull is represented as a balanced binary tree of the hull points, sorted by $x$-coordinate. Their algorithm repeated attempted to draw a tangent line from one hull to the other at the current two points we are considering. Depending on how that tangent line intersected the hulls, Overmars and van Leeuwen developed an analysis of ten cases determining how to recurse on each binary tree at each step. Each step always involves us going deeper in at least one of the trees, therefore the algorithm takes $O(\log n)$ time, assuming that both trees have depth $O(\log n)$. 

To apply path-guided pushdown random walks, Eppstein {\it et al.} considers the implicit decision tree created by the algorithm. Each node represents a state of the algorithm, and each child represents choosing one of ten cases to advance one of the two trees. This decision tree has depth $O(\log n)$ for the same reason the algorithm takes $O(\log n)$ time. They show that, using four invocations of Overmars and van Leeuwen's case analysis, we can verify if we are on the correct path of this decision tree. For more details, see~\cite{eppstein2025computational} (Section 6.2).

\section{Review of Non-noisy Algorithms}

In this section, we review the classic works on which we base our noise-tolerant algorithms.

\subsection{2D Convex Hull Algorithm for CREW PRAM.}\label{sec:2dhull-original}
Here we review the $O(\log n)$-span, $O(n\log n)$-work parallel convex hull algorithm independently discovered by Aggarwal, Chazelle, Guibas, Ó'Dúnlaing and by Goodrich and Atallah~\cite{goodrich1987efficient,aggarwal1988parallel}. Note that it is sufficient to describe a convex hull algorithm that computes the upper hull, as a symmetric algorithm can compute the corresponding lower hull of a set of points. 

We first sort all $n$ points by their $x$-coordinate in $O(\log n)$ span and $O(n\log n)$ work, e.g., using Cole's merge sort~\cite{cole1988parallel}. Then we partition them into $\sqrt n$ groups of $\sqrt n$ points. This partition is done by their index, e.g., points 1 through $\sqrt n$ in the sorted order belong in group 1, and so on. We recursively solve each group in parallel, and each group ultimately returns with an upper hull of their $\sqrt n$ points. 

Once we have $\sqrt n$ upper hulls, we can combine them in parallel. We first compute an upper tangent for every pair of upper hulls using the double binary search of Overmars and van Leeuwen~\cite{overmars1981maintenance}. 

We next determine the contribution of each upper hull $U_j$ to the combined upper hull of all $n$ points. To do this, we compute $V_j$, the tangent of smallest slope between $U_j$ and some other $U_i$ where $i < j$. Likewise, we compute $W_j$, the tangent of largest slope between $U_j$ and some other $U_k$ where $k > j$. Let $v_j$ and $w_j$ be the corresponding tangent points on $U_j$. If the angle between $V_j$ and $W_j$ around their intersection point pointing upward is $\geq 180\degree$, then points from $v_j$ to $w_j$ on $U_j$ are on the combined upper hull. Otherwise, $U_j$ contributes no points to the upper hull. 
The authors' rationale for choosing $V_j$ and $W_j$ 
is that $U_i$ and $U_k$ are the two hulls that are 
best able to ``cover up'' $U_j$ such that it does not contribute to the combined upper hull. Their respective tangents with $U_j$ indicate that they are the ``highest'' hulls on each side of $U_j$. Finally, a parallel prefix operation is used to combine contributions from each $U_j$ into the combined upper hull. 

Now we discuss the span and work performed at each level of recursion. There are $\binom{{\sqrt n}}{2} = O(n)$ pairs of upper hulls. The double binary search method of Overmars and van Leeuwen~\cite{overmars1981maintenance} takes $O(\log n)$ time sequentially. If we perform each tangent computation in parallel, this takes $O(\log n)$ span and $O(n\log n)$ work. To compute $W_j$, we perform a max-find operation on the $O(\sqrt n)$ tangents that touch $U_j$. Doing this in parallel over all $U_j$ takes $O(\log n)$ span and $O(n\log n)$ work. The same cost applies to computing $V_j$. It then takes constant span and $O(\sqrt n)$ work to determine the convex chain each $U_j$ contributes to the combined hull. The final parallel prefix operation of the $O(\sqrt n)$ convex chains takes $O(\log n)$ span and $O(\sqrt n)$ work. 

For $n \leq 2$, span and work are both constant as the convex hull is simply a line segment between the two points. We conclude that the total span and work for the algorithm are given by the following two recurrences:
$$T(n) = T(\sqrt n) + O(\log n)$$
$$W(n) = \sqrt n W(\sqrt n) + O(n\log n)$$

With $T(2), W(2) = O(1)$. Observe that both recurrences can be bounded by a geometric series. 
$$T(n) \leq O(\log n) \times \sum_{i=0}^\infty 1/2^i = O(\log n)$$
$$W(n) \leq O(n\log n) \times \sum_{i=0}^\infty 1/2^i = O(n\log n)$$

The initial sorting takes $O(\log n)$ span and $O(n\log n)$ work~\cite{cole1988parallel}, so the algorithm takes $O(\log n)$ span and $O(n\log n)$ work overall in the CREW PRAM model~\cite{aggarwal1988parallel,goodrich1987efficient}.

\subsection{3D Convex Hull Algorithm for CREW PRAM.}\label{sec:3dhull-original}
We adapt our algorithm mostly from the optimal $O(\log n)$-span and $O(n\log n)$-work randomized 3D convex hull algorithm of Reif and Sen~\cite{reif19923DHull}. 
First, we note that Reif and Sen solve the dual problem of halfspace-intersection of $n$ halfspaces in 3D. This allows them to take advantage of an earlier result of Clarkson~\cite{clarkson1988applications}, who showed that a random sample of halfspaces $R$ can be used to divide the remaining $H\setminus R$ halfspaces into subproblems of roughly equal size, admitting a natural divide-and-conquer algorithm. Specifically, these good samples divide the remaining halfspaces into groups with the following two properties: (1) the largest group has size at most $O(|H|\log |R|/|R|)$ and (2) the sum over all groups is of size $O(|H|)$. However, Clarkson only showed that this is the case with constant probability $> 1/2$. 

To turn this sampling probability into a high-probability bound, Reif and Sen introduce the notion of \emph{polling}. Rather than one sample $R$, they take $O(\log n)$ samples $R_j$ and test if any are ``good'' as defined by Clarkson. It would be too expensive to evaluate each $R_j$, so instead of sampling each from $H$, they first sample $n/\log^d n$ elements from $H$ into $H_j$, for $d > 2$. Then they sample $R_j$ from $H_j$. By shaving this polylog factor, Reif and Sen show that we can evaluate whether each $R_j$ is a good sample for $H_j$ in $O(\log n)$ span and $O(n\log n)$ work. Through applications of Chernoff bounds, they show that (1) at least one $R_j$ is good with respect to $H_j$, and (2) $R_j$ is a good sample for the original $H$ as well, despite the polylog factor decrease. Both occur with high probability in $n$. Note that if a good sample was not found, we could simply spend the same amount of span and work to repeat the process. 

Say that we have found a good sample $R$. In the process to determine whether $R$ was good, Reif and Sen use a non-optimal halfspace intersection algorithm to compute the polyhedron $\mathcal R$ in $O(\log n)$ span and $O(n\log n)$ work. We can use this structure to define our divide-and-conquer. Observe that the halfspace intersection $\mathcal C$ of $H$ can only get smaller than $\mathcal R$, thus the subproblems which we recursively solve must involve some partition of the interior of $\mathcal R$. Specifically, this is done by considering some point $o$ assumed to be in the halfspace intersection of $H$. We draw lines from each vertex of $\mathcal R$ to $o$ and triangulate the faces of $\mathcal R$ such that we have a set of at most $2|R|$ ``cones'', their a base a triple of vertices of $\mathcal R$ and their apex $o$. Then each recursive problem would involve computing a halfspace intersection of the subset of halfspaces in $H \setminus R$ that intersect some cone $C_i$.

To determine the cones each halfspace in $H\setminus R$ stabs, the authors consider the dual problem. Taking the dual transform of each vertex of $\mathcal R$ gives us an arrangement of hyperspaces. Each cell of the arrangement corresponds to a combination of cones $C_i$ that some halfplane in the primal could intersect. By parallelizing the multidimensional search method of Dobkin and Lipton~\cite{dobkin1976multidimensional}, they show that, as long as $R$ is a sufficiently small sample, $|R| \leq n^{1/8}$, we can construct a search structure in $O(\log n)$ span and $O(n\log n)$ work that allows for fast queries to return the set of cones of $\mathcal R$ that a hyperspace $h$ intersects in the primal in. Note that this sample guarantees that a subproblem can be as large as $O(n^{7/8}\log n)$.

Lastly, Reif and Sen observe that, despite our sample $R$ being ``good'', it is not good enough to prevent a blowup in problem size as we go down 
recursive levels. This is because halfspaces may intersect with multiple cones. If $X_i$ is the number of halfspaces of $H \setminus R$ that intersect cone $C_i$, then Clarkson showed that $\sum_i X_i = O(n)$~\cite{clarkson1988applications}, meaning the total problem size may increase by a constant factor every level of recursion. There are $O(\log\log n)$ levels of recursion, so this constant factor blow-up implies total problem size may expand to $O(n\log^{O(1)} n)$, increasing work by a polylog factor. Thus, their final step before recursion involves pruning redundant halfspaces to ensure that total problem size at each level of recursion is no more than $O(n)$.

In their work, Reif and Sen classify different types of redundancies and describe a two-step approach to pruning the halfspaces of cone $C_i$, first by applying a 3D maxima algorithm and then by computing a simplified 3D hull they call a skeletal hull. Information from these outputs can be used to narrow down only the halfspaces that both intersect $C_i$ and contribute a vertex to the final halfspace intersection $\mathcal C$ that lies in $C_i$. 

However, a later work by Amato, Goodrich, and Ramos~\cite{amato1994parallel} provides a simpler way to prune. They show that we can perform a sufficient amount of pruning from information given to us by the \emph{contours} of $C_i$. A contour is the result of projecting the halfspaces that intersect $C_i$ onto the three faces of the cone incident to the apex $o$ and then computing the 2D convex chain of the projection on each face. Given these contours, they show that a constant number of binary searches per halfspace and one sorting step per cone can be used to prune a sufficient number of halfspaces. By the fact that the total problem size among all cones is $O(n)$, we can compute contours for all cones with $O(\log n)$ span and $O(n\log n)$ work. In addition, Amato, Goodrich, and Ramos show that it takes $O(\log n)$ span and $O(n\log n)$ work to prune given these contours. 

Finally, we recursively solve each subproblem defined by each cone $C_i$. Each subproblem returns with the subset of vertices of the halfspace intersection $\mathcal C$ located in that subproblem, from which an adjacency structure can be computed by an application of sorting~\cite{reif19923DHull} (Lemma 2.2). Once the subproblem size is smaller than $O(\log ^k n)$ for some constant $k$, the authors can apply a non-optimal time algorithm (for example the $O(\log^3 n)$-span, $O(n\log^3 n)$-work algorithm of Aggarwal, Chazelle, Guibas, Ó'Dúnlaing, and Yap~\cite{aggarwal1988parallel}) to solve the problem directly without increasing total asymptotic complexity.
Reif and Sen conclude that the span for their algorithm is bounded by following recurrence:
$$T(n) \leq T(n^{9/10}) + O(\log n)$$

The $n^{9/10}$ comes from the analysis of Clarkson, who showed that with a good sample, no cone has more than $O(n^{7/8}\log n) = O(n^{9/10})$ halfspaces that stab it. Observe that the recurrence is the summation of a converging geometric sequence, so $T(n) = O(\log n)$. For the total work, we simply observe that the pruning step guarantees that each level of recursion has total problem size $O(n_0)$, where $n_0$ is the initial problem size. Given that the span of the algorithm at each level of recursion is $O(\log n)$, we can upper bound the work like so:
$$W(n) \leq W(n^{9/10}) + O(n_0) \times O(\log n)$$

We again have a geometric series, so $W(n_0) = O(n_0\log n_0)$.
We conclude that, given a set of $n$ halfplanes in 3D, Reif and Sen's algorithm computes their halfspace intersection in $O(\log n)$ span and $O(n\log n)$ work for the CREW PRAM model~\cite{reif19923DHull}.

\section{Deferred Details for Our 2D Convex Hull Algorithm}
\label{sec:2D-supp}

\subsection{Omitted Proofs}

Now we are ready to prove inductively that our algorithm succeeds with high probability in $n$ and, in doing so, find the appropriate constant $c$ determining the probability that an individual noisy operation fails.

\begin{theorem}[Same as Theorem~\ref{thm:2dc}]
    We can compute the convex hull of $n$ points in the plane in $O(\log n)$ span and $O(n\log n)$ work w.h.p. in $n$, even with noisy primitive operations.
\end{theorem}
\begin{proof}
    We will prove this by induction. Our inductive hypothesis is that we can compute a subproblem of size $m$ in $O(\log m)$ span and $O(m\log m)$ work with failure probability no more than $m^{-c'}$. When $m \leq 2$, this holds true because no noisy comparisons are needed to draw an edge between two points. 

    Consider a subproblem $S'$ of size $m^2$. To complete the proof, we will show that it can be solved in $O(\log(m^2))$ span and $O(m^2\log(m^2))$ work with failure probability no more than $m^{-2c'}$. $S'$ must combine $m$ subproblems of size $m$. We expect $m / m^{c'}$ to fail by our inductive hypothesis, yet we require that no more than $m^{1/3}$ fail by Corollary~\ref{lem:2D-bound-fail}. Let $X$ be a random variable representing the number of subproblems that fail. Markov's inequality states that $\Pr[X \geq m^{1/3}] \leq (m/m^{c'}) \times (1/m^{1/3}) = m^{2/3 - c'}$. To be a high probability bound in $m^2$, $2/3 - c' < -2$, so $c' > 8/3$.

    Continuing the inductive argument, we must show that $S'$ succeeds in outputting an upper hull of its points with failure probability no more than $m^{-2c'}$. Notice that every noisy max-find and upper tangent computation performed during brute-force reconstruction and the later combination step individually fails with probability no more than $m^{-2c}$. To determine the constant $c$, we count the number of noisy operations done and take a union bound over all of them to compute the probability that at least one fails. 

    By Lemma~\ref{lem:2D-failure-sweep}, sweeping for errors takes $O(m)$ noisy operations total, and each brute-force reconstruction takes $O(m)$ operations total. At most $m^{1/3}$ subproblems fail with high probability in $m^2$, so brute-force reconstruction requires at most $O(m^{4/3})$ noisy comparisons. Combining each subproblem $S$ into an upper hull over $S'$ requires $O(1)$ noisy max-find and upper tangent operations per point, so $O(m^2)$ operations. In total, we use $O(m^2)$ noisy operations to compute the upper tangent of $S'$. 

    Taking the union bound over all $O(m^2)$ invocations, if a single noisy operation fails with probability $m^{-2c}$, at least one fails with probability $\leq m^{-2(c - 1)}$. If we take $c \geq 1 + c' > 11/3$ and adjust constant factors associated with each operation accordingly, computing the upper hull for $S'$ takes $O(\log(m^2))$ span and $O(m^2\log(m^2))$ work and fails with probability $\leq m^{-2c'}$. 

    We conclude that our algorithm succeeds with high probability in $n$, and the total span and work for our algorithm are bounded by the following recurrences:
    $$T(n) = T(\sqrt n) + O(\log n)$$
    $$W(n) = \sqrt n W(\sqrt n) + O(n\log n)$$
    These are the same recurrences that bound the span and work of the original algorithms~\cite{goodrich1987efficient,aggarwal1988parallel}. They evaluate to $O(\log n)$ span and $O(n\log n)$ work respectively. 
\end{proof}

\subsection{Additional Details}
\label{sec:2D-supp-2}

\begin{figure*}[hbt]
    \centering
    \includegraphics[width=0.6\linewidth]{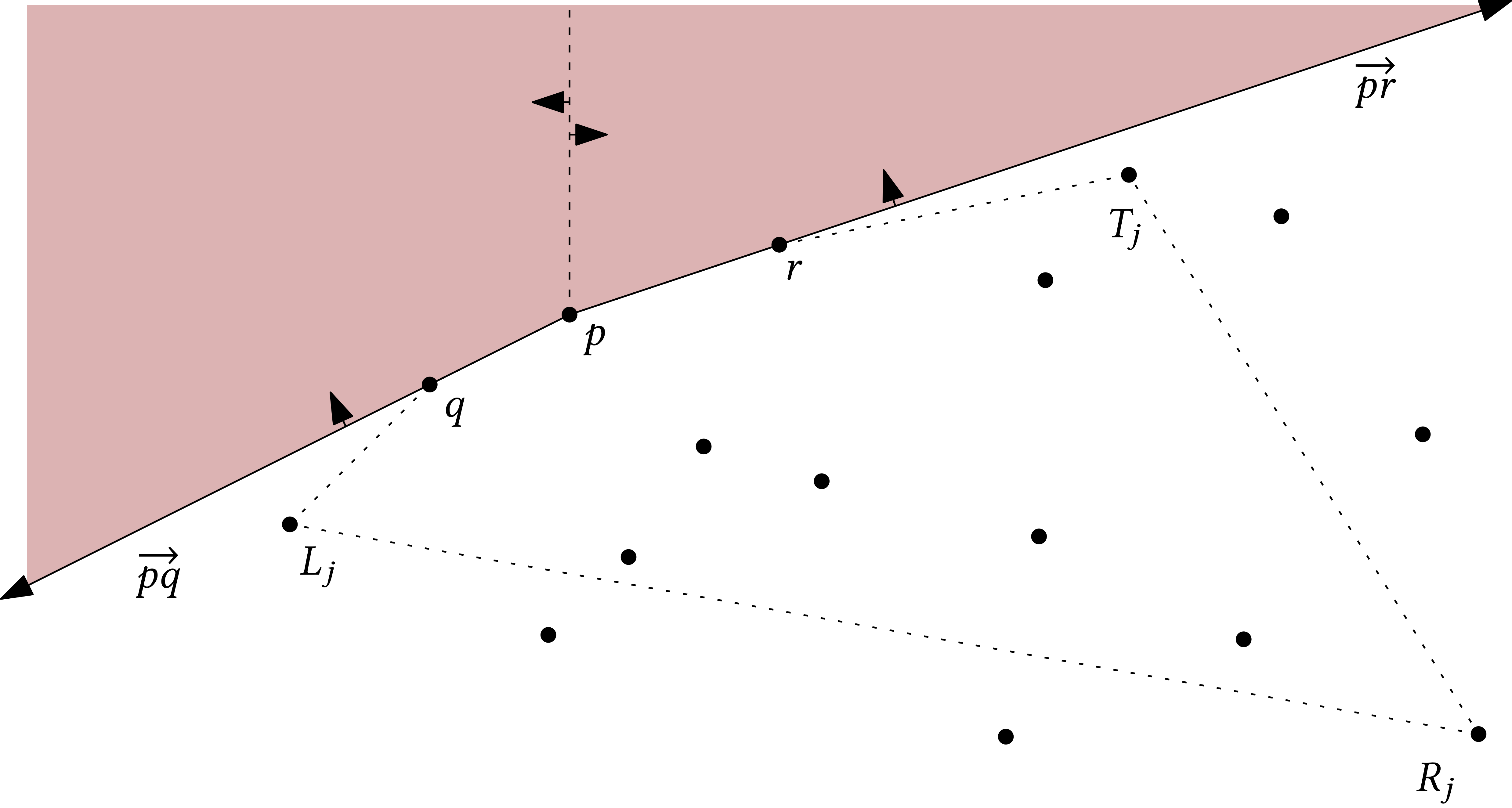}
    \caption{We depict an instance of the failure sweep brute-force
    algorithm determining that a point, $p$, is a part of the upper
    hull in $S_j$. In this case, $[L_j, q, p, r, T_j, R_j]$ is
    convex. We have depicted the lines $\protect\overrightarrow{pr}$ and
    $\protect\overrightarrow{pq}$ as well as the vertical line through $p$
    that borders both regions defined by the lines. Notice, as the
    proof describes, no points in $S_j$ can exist in the region
    either to the left of $\protect\overrightarrow{pr}$ and to the right
    of $\protect\overrightarrow{pq}$, depicted with shading. It is clear that $p$ is an extreme point in $S_j$.}
    \label{fig:brute-force-2dch}
\end{figure*}

\begin{lemma}[Same as Lemma~\ref{lem:2D-extreme}]
    \hphantom{}\\
    If either $[L_j, q, p, r, T_j, R_j]$ or $[T_j, q, p, r, R_j, L_j]$ is a convex chain, then $p$ is a member of the upper hull of $S_j$. 
\end{lemma}

\begin{proof}
By definition of CW and CCW neighbors, no points in $S_j$ exist in the region to the left of $\protect\overrightarrow{pr}$ and the the right of the vertical line through $p$. Likewise, no points in $S_j$ exist in the region to the right of $\protect\overrightarrow{pq}$ and the left of the vertical line through $p$. Because the two regions border each other and $[q,p,r]$ makes a right-hand turn, this region is the same as the union of the regions of the plane to the left of $\protect\overrightarrow{pr}$ and the right of $\protect\overrightarrow{pq}$. See Figure~\ref{fig:brute-force-2dch}.

In either of the two cases, imagine we rotate the coordinate system such that $p$ has the highest $y$-coordinate of all $\{T_j, R_j, L_j, p, q, r\}$. By the fact that no points exist to the left of $\protect\overrightarrow{pr}$ and to the right of $\protect\overrightarrow{pq}$, $p$ has the highest $y$-coordinate of all points in $S_j$ in this new coordinate system. Because the structure of a convex hull is invariant to rotation, $p$ must be on the convex hull of $S_j$. Returning to the original coordinate system, because either $[L_j, p, T_j]$ or $[T_j, p, R_j]$ is a right-hand turn and the fact that $L_j$, $T_j$, and $R_j$ are the leftmost, topmost, and rightmost points, $p$ must be on the upper hull of $S_j$, not its lower hull. 
\end{proof}

\begin{cor}[Same as Corollary~\ref{lem:2D-bound-fail}]
    During the recombination step of subproblem $S'$, at most $O(m^{1/3})$ subproblems can fail such that we take $O(m^2\log(m^2))$ work to recompute them. 
\end{cor}
\begin{proof}
    Say that $m^{1/b}$ subproblems fail. To recompute them in parallel takes $O((m^{1/b} \times m)^{3/2} \times \log(m^2))$
    $= O(m^{3/(2b) + 3/2}\log(m^2))$ work. We can only afford $O(m^2\log(m^2))$ work at subproblem $S'$, so $3/2b + 3/2 \leq 2$. We have that $b \geq 3$, so at most $O(m^{1/3})$ subproblems can fail such that $S'$ can fix them in $O(m^2 \log(m^2))$ work via the brute-force method described in the previous section.
\end{proof}

\section{Deferred Details for Our 3D Convex Hull Algorithm}\label{sec:3D-supp}

\subsection{Adapting the Non-Noisy Algorithm to the Noisy Primitives Setting}\label{sec:3D-overview}

We begin with an explanation of how to modify the non-noisy algorithm such that, under noisy primitive operations, each step succeeds with high probability in the size of the current subproblem. From an application of the trivial repetition strategy, we immediately have the following algorithms to compute a halfspace intersection. Respectively, these are used to compute the halfspace intersection of our sample $R$ and to directly solve sufficiently small subproblems.

\begin{lemma}[Same as Lemma~\ref{lem:3D-brute-force}]
    The intersection of a set of $n$ halfspaces can be computed in $O(\log n)$ span and $O(n^4\log n)$ work with high probability in $n$ under noisy comparisons.
\end{lemma}
\begin{proof}
    Follows from Lemma 2.2 of Reif and Sen~\cite{reif19923DHull}. There are $O(n^3)$ candidate vertices in a 3D arrangement and $O(n)$ halfspaces, so there are $O(n^4)$ vertex-halfspace pairs. For each vertex-halfspace pair, we can determine whether the vertex satisfies the halfspace's equation in a single noisy operation. For each vertex, we perform a parallel prefix operation to compute the conjunction of the $O(n)$ results. 

    There are $O(n^4)$ pairs, so testing each with the trivial repetition strategy takes $O(n^4\log n)$ work. Each computation can be done in parallel, so this step takes $O(\log n)$ span. The final prefix operations done at each vertex are not noisy. They take $O(\log n)$ span and $O(n^4)$ work overall. 
    
    We are left with a set of vertices in the arrangement of the given halfspaces that satisfy every halfspace. This means these vertices must bound the halfspace intersection. An application of noisy sorting~\cite{goodrich2023optimal} can be used to compute an adjacency structure for the polytope from this set in $O(\log n)$ span and $O(n\log n)$ work~\cite{reif19923DHull}. 

    We perform $O(n^4)$ noisy operations in this step. If each individually fails with probability no more than $1/n^{c}$ for some constant $c$, then by union bound at least one fails with probability $1/n^{c - 4}$. Therefore, to guarantee a failure probability of at most $1/n^{c'}$, we set $c > 5 + c'$. This only affects constant factors in the work and span of these operations and so does not change their overall complexity.
\end{proof}

\begin{lemma}\label{lem:3D-nonoptimal}
    The intersection of a set of $n$ halfspaces can be computed in $O(\log^4 n)$ span and $O(n\log^4 n)$ work with high probability in $n$ under noisy comparisons. 
\end{lemma}
\begin{proof}
    Follows from the $O(\log^3 n)$-span, $O(n\log^3 n)$-work algorithm of Aggarwal, Chazelle, Guibas, Ó'Dúnlaing, and Yap~\cite{aggarwal1988parallel} and the trivial repetition strategy applied to every noisy primitive. There are certainly at most $O(n\log^3 n) = O(n^2)$ primitive operations performed throughout the algorithm. We can replace them all with the trivial repetition strategy. Say that we configure each noisy operation to fail with probability at most $n^{-c}$. Then, by union bound, one of these comparisons fails with probability at most $n^{-(c-2)}$. If we wish the algorithm to fail with probability no more than $1/n^{c'}$, we set $c \geq 2 + c'$.
\end{proof}

We next observe that Reif and Sen's \emph{polling} method~\cite{reif19923DHull} in and of itself is not subject to noise as randomly selecting halfspaces requires no primitive operations. However, determining whether a sample is good does involve noisy operations. We now show that computing $\mathcal R$, the halfspace intersection of $R$, and the point-location structure that determines what cones a halfspace in $H\setminus R$ intersects can be performed in $O(\log n)$ span and $O(n\log n)$ work w.h.p.

\begin{lemma}
    Given a point $o$ in the halfspace intersection of $H$, we can compute the halfspace intersection of $R$ and its cones using $O(|R|^4\log n)$ work and $O(\log n)$ span under noisy comparisons w.h.p.
\end{lemma}
\begin{proof}
    Follows from Lemma~\ref{lem:3D-brute-force}. The adjacency structure of $R$ that we compute is correct with high probability in $n$. Because we have computed this structure, any topological changes we make to produce the cones are not noisy. Thus, we can partition $\mathcal R$ with planes such that each face is a trapezoid and partition each trapezoid further into a triangle, and draw edges between each face and our origin $o$ as Reif and Sen describe~\cite{reif19923DHull} without noise. 
\end{proof}
\begin{lemma}\label{lem:dobkins-structure}
    We can preprocess $n$ planes into an arrangement such that, given an arbitrary query point, we can report the cell in which the point is contained. Preprocessing can be done under noisy comparisons in $O(n^7\log n)$ work and $O(\log n)$ span w.h.p. using $O(n^7)$ noisy operations. Queries can be performed sequentially in $O(\log n)$ time w.h.p. using $O(1)$ noisy operations.
\end{lemma}
\begin{proof}
    This follows from~\cite{reif19923DHull} (Lemma 5.1) as well as by the work of~\cite{dobkin1976multidimensional}. Each preprocessing step involves projection of points down to a lower dimension and a constant number of parallel sorting steps per cell. Projection is not a noisy operation as we are manipulating a geometric object's equations, which are its labels. Each sorting step can be swapped with the noisy sorting algorithm of~\cite{goodrich2023optimal}. Likewise, each query involves a constant number of binary search steps, which can be replaced with the noisy binary search implementation of~\cite{feige1994computing}. 
\end{proof}
\begin{cor}
    We can compute a preprocessing structure from $\mathcal R$ such that, given a query halfspace $h$, we can determine which cones in $\mathcal R$ $h$ stabs. Preprocessing can be performed in $O(|R|^8\log n)$ work and $O(\log n)$ span w.h.p. using $O(|R|^8)$ noisy operations. Queries can be performed sequentially in $O(\log n)$ time w.h.p. using $O(1)$ noisy operations.    
\end{cor}
\begin{proof}
    Follows from our Lemma~\ref{lem:dobkins-structure} and~\cite{reif19923DHull} (Lemma 5.2) applied to the dual of each vertex of $\mathcal R$. There are $O(|R|^7)$ cells of our arrangement and $O(|R|)$ cones in $\mathcal R$. We have each of the $O(|R|^8)$ cone-cell pairs compute whether the cell overlaps the cone in $O(1)$ noisy comparisons. $O(|R|^7)$ parallel prefix operations can be used to compile a list of cones each cell is incident to.

    Testing each cone-cell pair takes $O(|R|^8\log n)$ work and $O(\log n)$ span via the trivial repetition strategy. The parallel prefix operations take $O(|R|^8)$ work and $O(\log |R|) = O(\log n)$ span.

    As described in the previous lemma, each query requires a constant number of noisy binary searches and can be accomplished in $O(\log n)$ time sequentially w.h.p.
\end{proof}
Setting $|R| \leq n^{1/8}$ allows us to perform the previous steps in $O(\log n)$ span and $O(n\log n)$ work w.h.p. Once the structure is built, we perform one query per halfspace $h \in H \setminus R$ to determine which cones it stabs. All queries can be done in parallel and there are $O(n)$ halfspaces, so this step takes $O(\log n)$ span and $O(n\log n)$ work. Having done this work, we can use this data structure to determine if our current sample $R$ is ``good'', i.e. if it distributes halfplanes roughly evenly between cones~\cite{clarkson1988applications,reif19923DHull}. 

\subsubsection{Polling work and span bounds}
Before we continue, we note that Reif and Sen~\cite{reif19923DHull} show that we can adjust parameters such that each instance of polling at recursive depth $l$ fails with probability at most $n^{-(9c/10)^l}$ for some $c > 1$. Let $n_l = n^{(9/10)^l}$, for simplicity. They observe that polling takes work $O(n_l \log n_l)$ and span $O(\log n_l)$. If an instance of polling fails, we can simply repeat the process again. To show that polling takes no more than $O(\log n)$ span overall, Reif and Sen prove the following useful theorem.

\begin{theorem}\label{lem:span-thm}
    If a subproblem at recursive depth $l$ takes $O(\log n_l)$ span w.h.p. in $n_l$, then the entire problem takes $O(\log n)$ span. 
\end{theorem}
\begin{proof}
    See~\cite{reif19923DHull}~(Theorem 2.1)~and~\cite{reif1992optimal}~(Theorem 2).
\end{proof}

A work bound follows like so.

\begin{cor}\label{lem:work-thm}
    If total work done by all subproblems at recursive level $l$ is $O(n\log n_l)$ w.h.p. in $n_l$, then the entire problem takes $O(n\log n)$ work w.h.p. in $n$. 
\end{cor}

In our algorithm, the probability that polling succeeds is the probability that a good sample is found and all noisy operations used to test each sample is correct. Say the probability polling succeeds is $1/n_l^\alpha$ and the probability each individual noisy operations succeeds is $1/n^c$. Because $O(n_l)$ noisy operations are performed to test all samples, by union bound, at least one fails with probability $1/n_l^{c-1}$. Assuming $c > 2$ (requires an adjustment of constant factors for noisy operations), both of these are high-probability bounds. 

Then the probability either of these fails is at most $1/n_l^\alpha + 1/n_l^{c-1}$, which is less than $2\times \max\{1/n_l^\alpha, 1/n_l^{c-1}\}$. This is still a high-probability bound in $n_l$, so we can invoke Theorem~\ref{lem:span-thm} to bound the span of polling even in the presence of noisy primitives. 

In the next section, we describe how to prune halfspaces w.h.p. in the size of the current subproblem such that each level of recursion has total problem size $O(n)$. With a similar argument to above we can use Corollary~\ref{lem:work-thm} to bound total work of polling to $O(n\log n)$ w.h.p.

\subsubsection{Pruning halfspaces}
To ensure total problem size among subproblems at a given level of recursion does not exceed some constant multiple of our initial problem size, we rely on a method of Amato, Goodrich, and Ramos~\cite{amato1994parallel}. 
\begin{lemma}\label{lem:pruning-runtime}
    In a problem size of $n$, where $H$ is the total set of halfplanes, we can prune enough redundant halfspaces such that total problem size among all subproblems at this level is at most $O(n_0)$, where $n_0$ is the initial problem size. This can be done in $O(\log n)$ span and $O(n\log n)$ work w.h.p. in $n$ using $O(n)$ noisy operations.
\end{lemma}
\begin{proof}
    Follows from the pruning strategy and corresponding lemmas described in~\cite{amato1994parallel} (Section 4.1.1). Their solution first involves projecting the halfspaces associated with $C_i$ onto the three faces of the cone, which is not a noisy operation. Then they compute the 2D convex hull of each projection, which we can do in $O(\log n)$ span and $O(n\log n)$ work w.h.p. in $n$ using our construction in Section~\ref{sec:our-2D-hull}. They call these convex chains \emph{contours}.

    They consider separately the halfspaces that contribute an edge to a contour and those that do not. We denote the set of halfspaces that stab $C_i$ but do not contribute to any of its contours $H_{|C_i}^{nc}$. Let $H^{nc}$ be the union of all such halfspaces over all cones. The authors show that each halfspace $h \in H^{nc}$ may contribute a vertex to the halfspace intersect of at most one cone~\cite{amato1994parallel} (Lemma 4.2). They prove that this halfspace $h$ can be found by locating the closest point $p$ on each of the three contours of $C_i$ and shooting a ray from each $p$ through $C_i$ towards $h$. If all three such rays pierce $h$ without being stopped first by their respective halfspaces $h', h''$ that contributed $p$ to the contour, then $h$ may contribute a vertex to the halfspace intersection of $H$. 

    Each halfspace $h \in H_{|C_i}^{nc}$ can determine whether it should be pruned in $O(\log n)$ time w.h.p. by projecting itself onto each face and performing a noisy binary searches on each contour to determine the closest hull point to the projected line representing $h$ (recall that we maintain the upper and lower hulls of each contour as a binary search tree, something we took advantage of to perform failure sweeping in Section~\ref{sec:2D-failure-sweep}). We can associate each contour point $p$ with the two halfspaces that define it so that they can be recovered in constant time. We can then perform a constant number of noisy comparisons using the trivial repetition strategy to determine which halfspace the ray passes through first in $O(\log n)$ time w.h.p. By the fact that our sample $R$ is good, $|H^{nc}| = O(n)$. Then we can perform each test in parallel in $O(\log n)$ span and $O(n\log n)$ work w.h.p. 

    Now we consider the halfspaces that contribute to contours, which we denote as $H_{|C_i}^{c}$. Amato, Goodrich, and Ramos~\cite{amato1994parallel} show that, if a halfspace in $h \in H_{|C_i}^c$ contributes an edge to at least two of the contours, then it can be pruned. To determine which halfspaces should be pruned, we simply label each contour edge with the corresponding halfspace that generated it, which can be done when constructing the contours, and sort the labels lexicographically. Again, by our use of a good sample, the total number of contour halfspaces $|H^{c}| = O(n)$, thus we can perform noisy sorting in parallel over all $C_i$ in $O(\log n)$ span and $O(n\log n)$ work w.h.p.
    In the lexicographically sorted array, if a halfspace is adjacent to itself, then it must contribute an edge to at least two contours, and so should be pruned. 
    In~\cite{amato1994parallel}, the authors show that after pruning, the set of all contour halfspaces across all cones and all subproblems also has size $O(n_0)$. 
\end{proof}
Finally, to complete the halfspace intersection before returning from a subproblem, Amato, Goodrich, and Ramos note that we must add back a certain subset of halfspaces in $H_{C_i}^c$ that contributed an edge to the halfspace intersection but not a vertex. Fortunately, they show that at most three such halfspaces exist and can be identified when pruning. We can incorporate these halfspaces like so.

Let $\mathcal C_i$ be the halfspace intersection of cone $C_i$. Process each halfspace one at a time. Have each vertex in $\mathcal C_i$ check if it satisfies the halfspace's equation in $O(\log n)$ time using the trivial repetition strategy. Label the vertex based one whether it succeeded or failed. Then examine each edge. For edges that have one vertex included in $h$ and the other vertex not included in $h$, compute the intersection between $h$ and the edge to produce a new set of vertices that are incident to $h$. 

For each of the three halfspace processed, this takes $O(\log n)$ span and $O(|\mathcal C_i|\log n)$ work. Because 3D hulls have complexity linear in their input and $|H^{nc}|, |H^c| = O(n)$, this takes $O(n\log n)$ work over all cones. 

We conclude that each operation in a non-noisy 3D hull algorithm can be replaced with noisy operations without increasing the asymptotic complexity of each subproblem. We observe that, for a subproblem of size $n$, all operations discussed prior use at most $O(n)$ noisy operations, including sorting, searching, and trivial repetitions.

\subsection{Failure Sweeping}

We have shown that each step of the algorithm of~\cite{reif19923DHull}, with the pruning method of~\cite{amato1994parallel}, can be replaced with noise-tolerant versions that individually succeed with high probability in the current problem size. However, as we did for 2D convex hulls, we must negotiate the fact that our success probabilities get weaker as problem size gets smaller. To compensate for this, we once again perform failure sweeping. Unlike for 2D convex hulls, subproblems may not be evenly sized. We will denote $m$ as the size of the problem at recursive level $l$ and $M$ the size of its parent subproblem at recursive level $l-1$.

Once again, we proceed in two steps. We first determine if the provided polyhedron is convex and then determine whether it is a valid halfspace intersection, both w.h.p. in $M$. By general position, each vertex of this polyhedron has at most three edges incident to it and so has at most three neighboring vertices. We can verify that the polyhedron is convex by having each vertex perform one noisy 3D orientation test against its neighbors using the trivial repetition strategy. This takes $O(\log M)$ time per vertex to succeed w.h.p. in $M$. There are $m$ vertices, so total work done is $O(m\log M)$ and span is $O(\log M)$. A prefix operation taking $O(m)$ work and $O(\log m)$ span can combine results over all $O(m)$ vertices and determine whether the intersection is convex or not. 

Now we must ensure that this polyhedron satisfies all halfspaces assigned to this subproblem. It will be simpler to consider the equivalent dual problem: given a 3D convex hull, verify that all points on the hull are convex with respect to their neighbors and all other points lie within the hull. We have verified the first condition with our initial tests in the primal. When we performed failure sweeping on a 2D hull, we were able to quickly determine which edge of the hull lies directly above each point through noisy binary search on the tree-based structure that represented the hull. In this algorithm, however, no equivalent structure is built. As a result, we will show how to reduce the problem to planar point location and then describe an algorithm that can construct a PPL data structure efficiently under noisy comparisons. 

Consider a coordinate system in which the origin lies outside of the hull, and consider the plane $p$ formed by the $x$- and $y$-axes. Let the ``upper hull'' be the portion of the hull that can see this plane and the ``lower hull'' be the portion that cannot. This test can be done via an constant number of primitive operations per face, and so can be performed in parallel in $O(\log n)$ span and $O(n\log n)$ work using the trivial repetition strategy. 

We focus on the upper hull, as the case for the lower hull is symmetric. We can project each point of the upper hull simply by ignoring its $z$-coordinate. By general position, each face of our hull is a triangle. Therefore, this projection is a triangulation of a convex polygon. We can use this projection to determine if some non-hull point $u$ lies below the upper hull. First, we project $u$ onto $p$, creating the projected point $u^*$. We use planar point-location to determine the triangle $\Delta^*$ it appears in. This triangle maps to a face of the upper hull, $\Delta$, that intersects the ray $\protect\overrightarrow{uu^*}$. Next, we simply perform an orientation test between $u$ and $\Delta$ to determine if $u$ is inside or outside of the upper hull. We can repeat this point location test symmetrically on the lower hull. If any $u$ fails either test, then we conclude that our convex hull is invalid. 

\begin{figure*}
    \centering
    \includegraphics[width=0.45\linewidth]{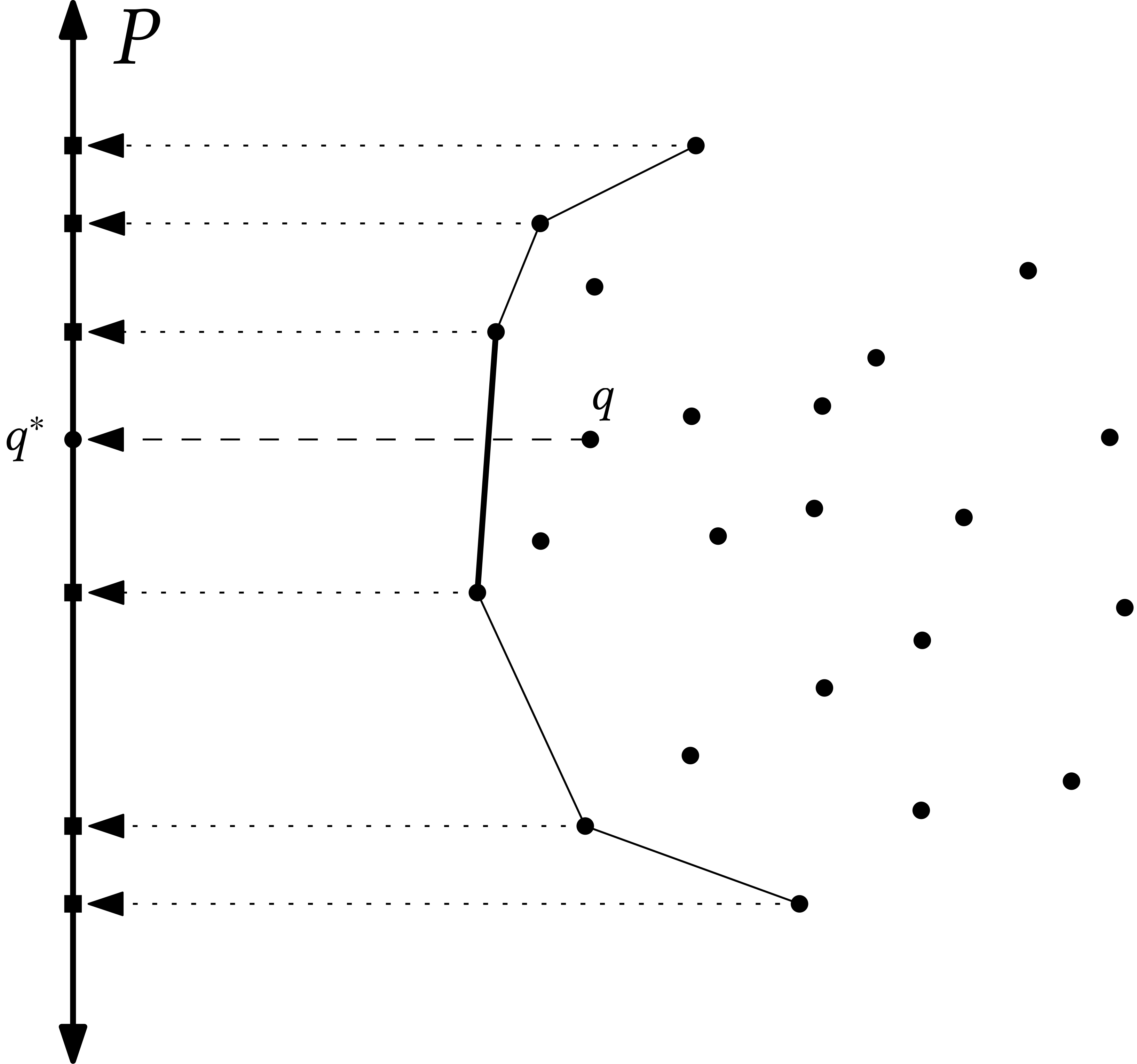}
    \caption{Here we illustrate our approach to failure sweeping a 3D hull. We define a coordinate system with its origin outside of the given hull and let $P$ be the $y$-axis (the $x$-$y$ plane in 3D). We define the upper hull as the hull points that can see $P$. We then perform an orthogonal projection of the upper hull points onto $P$, inducing a set of intervals. In 3D, this would induce a trianglated polygon onto $P$. We can project a query point $q$ onto $P$ and determine what region of $P$ $q^*$ lies in. This tells us what portion of the hull $\protect\overrightarrow{qq^*}$ stabs. In 2D, this requires a single binary search. However, in 3D we require a planar point-location structure to determine the triangle $q^*$ lies in. }
    \label{fig:3D-failure-sweep}
\end{figure*}

\subsubsection{Constructing a Planar Point-Location Data Structure on a Triangulated Convex Polygon} 
While the projections done to produce our triangulated convex polygon on plane $p$ are without noise, the work needed to construct the planar point location data structure does involve noisy primitive operations. 

To build a PPL data structure, we modify the randomized parallel algorithm of Reif and Sen~\cite{reif1992optimal} that constructs a Kirkpatrick decomposition~\cite{kirkpatrick1983optimal} of a triangulated polygon of $n$ points taking $O(\log n)$ span and $O(n\log n)$ work with high probability in $n$. Their basic strategy is to, at each step, have each vertex of degree $\leq d=O(1)$ independently flip a coin. Vertices that flip heads and have no neighbors that also flip heads are pruned. The gaps left behind are re-triangulated. While constructing the decomposition, Reif and Sen at the same time construct a point-location search DAG. Each triangle is represented by a node in the DAG. Triangles destroyed in a given iteration have a constant number of new triangles that overlap them. Nodes for each new triangle are given a pointer to every triangle destroyed at this step that they overlap with. Each of these operations takes $O(1)$ time in the non-noisy setting.

They prove that, over $O(\log n)$ steps, we remove a constant fraction of vertices at each step such that the last level of the decomposition is of size 
$O(\log n)$ with high probability in $n$.
Because there are $O(\log n)$ steps, each takes $O(1)$ time per point, and there are $O(n)$ points, the algorithm takes $O(n\log n)$ work and $O(\log n)$ span. They also observe that the depth of the resulting search structure is $O(\log n)$~\cite{reif1992optimal,kirkpatrick1983optimal}, allowing for $O(\log n)$-time sequential queries.

To convert this into a noise-tolerant algorithm, we begin by discussing how to introduce a bounding triangle to our triangulated convex polygon. We introduce our bounding points $a,b,c$. To triangulate the boundary of our polygon, we assign each hull point of our polygon a processor. We consider one bounding point at a time. Using $O(1)$ noisy orientation tests, each point on our polygon's hull can determine whether it ``sees'' bounding point $a$, and if so can draw an edge to it. It is easy to see that none of these edges cross. This takes $O(\log n)$ time per point using the trivial repetition strategy.

Now we consider $b$. Say that point $p$ already has an edge to $a$. If $p$'s neighbor in the direction of $b$ also has an edge to $c$, then $p$ should not draw an edge to $b$ as it will cross its neighbor's edge to $a$. Otherwise, $p$ draws an edge to $b$. Again, this takes $O(\log n)$ time per point using the trivial repetition strategy. Incorporating $c$ involves a similar process. There are $O(n)$ points $p$ and each can be processed independently at each of the three steps, so this takes $O(\log n)$ span and $O(n\log n)$ work.



After this, at each iteration, we have three steps: (1) sampling the independent set, (2) re-triangulating polygons, and (3) updating the search DAG. 
We first observe that the sampling step is without noise as we are just generating, reading, and comparing labels on points. For the re-triangulating step, we can remove noise via an initial noisy sorting computation on the $x$-coordinates of all $O(n)$ points in $O(\log n)$ span and $O(n\log n)$ work~\cite{feige1994computing,goodrich2023optimal}. 
If we label each point with its position in the sorted order, a single processor can re-triangulate a polygon of constant size in $O(1)$ time by comparing each point's labels with the others to determine in what order to sweep the polygon.

Lastly, we have to update the search DAG, which requires us to determine whether a destroyed triangle intersects with a constant number of new triangles. Because the triangles are created throughout the algorithm, it does not seem that we can preprocess the points to remove noise from these operations. Thus, if we were to construct the DAG as we compute the decomposition, the trivial repetition strategy would require that each of the $O(\log n)$ steps take $O(\log n)$ span. Instead, we apply a two-pass system. 

On the first pass, construct the Kirkpatrick decomposition and create a dummy node for each triangle. We can associate each dummy node with a triangle of the decomposition and with the triangles that replaced it with a constant number of links per node. On the second pass, consider each triangle and their corresponding DAG node in parallel. For a triangle destroyed at step $i$, it determines which of the at most $d = O(1)$ triangles created in the corresponding gap at step $i$ intersect it. This requires $O(1)$ primitive operations, and so can be performed in $O(\log n)$ time using the trivial repetition strategy. Each triangle then draws edges from the new triangles that overlap it to itself, creating our search structure. This takes $O(\log n)$ time per triangle and there are at most $O(n)$ triangles in the search structure~\cite{kirkpatrick1983optimal}, so this takes $O(\log n)$ span and $O(n\log n)$ work.

\subsubsection{Querying the Planar Point-Location Data Structure Efficiently}

As we described earlier, Reif and Sen's algorithm stops when the decomposition is of size $O(\log n)$~\cite{reif1992optimal}. This is sufficient in the non-noisy case as a query point can simply brute-force search through all triangles at this first level in $O(\log n)$ time to determine which node of the DAG it should start at. Here, however, the trivial repetition strategy would add an extra log factor to this brute-force search. The following lemma will help us construct another point-location structure to help us find the initial triangle in $O(\log n)$ time w.h.p. in $n$. One can think of this approach as a brute-force parallelization of the planar point location methods of Cole~\cite{cole1986searching} and Sarnak and Tarjan~\cite{sarnak1986planar}.
\begin{lemma}\label{lem:trivial-PPL}
    Consider a planar subdivision of size $r \leq n$. Using $O(\log n)$ span, $O(r^3\log n)$ work, and $O(r^2)$ space, we can construct a point-location data structure with high probability in $n$ that allows for $O(\log n)$-time queries w.h.p. under noisy primitives. Construction requires $O(r^3)$ noisy operations and querying requires $O(1)$ noisy operations.
\end{lemma}
\begin{proof}
    We begin by drawing a vertical line at each point, subdividing the plane into $r+1$ slabs. Each slab has at most $O(r)$ regions in it, so there are at most $O(r^2)$ regions. There were $O(r)$ initial triangles, so there are $O(r^3)$ triangle-region pairs. For each pair, we can determine whether the triangle and region intersect. This requires $O(1)$ noisy comparisons per pair. A parallel prefix operation can be used to combine all $O(r)$ outcomes for each region to determine the unique triangle that corresponds to it. In total, this takes $O(r^3\log n)$ work and $O(\log n)$ span for each operation to succeed w.h.p. in $n$.

    We can sort the slabs by their $x$-coordinate in $O(r\log n)$ work and $O(\log n)$ span w.h.p. in $n$. We can then sort each slab's regions vertically in $O(r^2\log n)$ total work and $O(\log n)$ span w.h.p. in $n$.

    Given a query point $q$, we can use two noisy binary searches to locate the region in which $q$ lies. The first determines the appropriate slab and the second determines the appropriate region within the slab. 
\end{proof}

In our case, $r= O(\log n)$, so total work is $o(n\log n)$ and extra space usage is $o(n)$. Lastly, it remains to show how a processor can navigate the search DAG in $O(\log n)$ time w.h.p. in $n$. To do this, we will construct a transition oracle for the Kirkpatrick decompositon DAG such that we can apply path-guided pushdown random walks. First, we observe that a given query point must have a unique valid path in the Kirkpatrick decomposition. Level $i$ of the DAG corresponds to triangulation $G_i$ of the decomposition. By definition of a triangulation, no triangles overlap, meaning that there is a unique triangle at level $i$ that corresponds to $q$. Therefore, there is a unique node at level $i$ of the DAG that is valid for $q$. This is true at all levels of the DAG, so each $q$ has a single valid path through the DAG. 

It also follows that we can determine whether we are on the correct path using a single noisy primitive. We simply check whether $q$ is located inside the triangle represented by our current node in the DAG. Because the Kirkpatrick decomposition search DAG satisfies these requirements,~\cite{eppstein2025computational} show that a single query point can navigate the DAG in $O(\log n)$ time w.h.p. in $n$ by instantiating path-guided pushdown random walks (see Appendix~\ref{sec:random-walks}). We conclude the following:

\begin{lemma}\label{lem:planar-point-location}
    Given a triangulation of $n$ points in the plane, we can construct a point-location data structure to determine which triangle a query point lies in under noisy primitives. It takes $O(\log n)$ time to construct the data structure and $O(\log n)$ time to query it w.h.p. in $n$. Construction uses $O(n)$ noisy operations. Querying uses $O(1)$ noisy operations.
\end{lemma}

\begin{cor}
    We can perform failure sweeping on a halfspace intersection of $m$ points in $O(m\log M)$ work and $O(\log M)$ span w.h.p. in $M$. This uses $O(m)$ noisy operations. 
\end{cor}
\begin{proof}
    Follows from Lemma~\ref{lem:planar-point-location} and the above discussion. Using the point-location data structure instantiated with failure probabilities in $M$, we can project each point $q$ to the $xy$-plane, and determine which triangle the ray $\protect\overrightarrow{qq^*}$ stabs in $O(\log M)$ time per point. From there we can perform one noisy orientation test to determine if $q$ lies inside or outside the hull, which takes $O(\log M)$ time using the trivial repetition strategy. There are $O(m)$ points, so this takes $O(m\log M)$ work and $O(\log M)$ span.
\end{proof}

\begin{theorem}[Same as Theorem~\ref{lem:failure-sweep-3D}]
    Any parent subproblem of size $M$ can failure sweep all of its child subproblems in $O(M\log M)$ work and $O(\log M)$ span w.h.p. in $M$. This uses $O(M)$ noisy operations. 
\end{theorem}
\begin{proof}
    Follows from above done in parallel over all subproblems. By Lemma~\ref{lem:pruning-runtime}, the total size of all subproblems is $O(M)$. 
\end{proof}

\subsection{Deferred Proofs}\label{sec:3D-supp-pf}
\begin{theorem}[Same as Theorem~\ref{lem:3D-final-thm}]
    We can compute the intersection of $n$ halfspaces in $O(\log n)$ span and $O(n\log n)$ work w.h.p. in $n$, even with noisy primitive operations.
\end{theorem}
\begin{proof}
    We will prove this by induction. Say we are at a parent subproblem of size $M$. Our inductive hypothesis assumes that we can compute each of its child subproblems of size $m$ with failure probability at most $M^{-c'/2}$ in $O(\log M^{1/2})$ span and $O(m\log M^{1/2})$ work. 

    Our Section~\ref{sec:3D-overview} shows that the remaining steps to split and recombine subproblems take $O(\log M)$ span and $O(M\log M)$ work in total. From our discussion above, it takes $O(\log M)$ span and $O(M\log M)$ work to perform size reduction and the corresponding reconstruction after returning from the recursive calls w.h.p. in $M$. Lastly, by Theorem~\ref{lem:failure-sweep-3D}, Theorem~\ref{lem:brute-force-thm}, and our discussion of size reduction, it takes $O(\log M)$ span and $O(M\log M)$ work to failure sweep and recompute all failed subproblems w.h.p. in $M$. This holds as long as $c' > 21/8$.

    It remains to determine the failure parameter $c$ for each individual primitive operation. 
    As noted at the end of Section~\ref{sec:3D-overview}, sampling, constructing the sampled hull and point location structure, and pruning all take $O(M)$ noisy operations in total. 
    Our analysis above also implies that size reduction overall also takes $O(M)$ noisy operations.
    By Theorem~\ref{lem:failure-sweep-3D} and Theorem~\ref{lem:brute-force-thm}, it takes $O(M)$ total noisy operations to failure sweep all problems and recompute any failed subproblems by brute-force. By union bound over all $O(M)$ invocations of noisy primitives, if a single noisy primitive fails with probability $\leq M^{-c}$, at least one fails with probability $O(M^{-(c-1)})$. Then we can take $c \geq c' + 1 > 29/8$ to achieve the desired confidence bound. 

    We have shown that failure sweeping in combination with our new size reduction technique allow us to efficiently upgrade the success probabilities of subproblems to that of their parents. We conclude by analyzing total span. Span is bounded by the following recurrence:
    $$T(n) = T(\sqrt n) + O(\log n),$$
    which evaluates to $O(\log n)$. 
    Total work is proportional to problem size times span at each level of recursion. By pruning, each level of recursion has total problem size at most $O(n_0)$, where $n_0$ is the initial problem size. Then work is bounded by the following recurrence:
    $$W(n) = W(\sqrt n) + O(n_0\log n).$$
    This sums to $O(n_0\log n_0)$. 
    We conclude that it takes $O(\log n)$ span and $O(n\log n)$ work w.h.p. to compute all the intersection of $n$ halfspaces under noisy primitives.

\end{proof}


\end{document}